\documentclass[3p]{elsarticle}
\usepackage{amsmath}
\usepackage{amssymb}
\usepackage{amsfonts}
\usepackage{subfigure}
\usepackage{graphicx}
\usepackage[british]{babel}

\newtheorem{observation}{Observation}

\newtheorem{theorem}{Theorem}
\newproof{proof}{Proof}

\begin{document}

\raggedbottom

\title{Locating a tree in a phylogenetic network}
\author[uoc]{Leo van Iersel\fnref{fn1}}
\ead{l.j.j.v.iersel@gmail.com}
\author[uoc]{Charles Semple\corref{cor1}\fnref{fn1}}
\ead{c.semple@math.canterbury.ac.nz}
\author[uoc]{Mike Steel\fnref{fn1}}
\ead{m.steel@math.canterbury.ac.nz}
\cortext[cor1]{Corresponding author}
\fntext[fn1]{We thank the Allan Wilson Centre for Molecular Ecology and Evolution and the New Zealand Marsden Fund for funding.}
\address[uoc]{Department of Mathematics and Statistics, University of Canterbury\\Private Bag 4800, Christchurch, New Zealand}

\begin{abstract} \emph{Phylogenetic trees} and \emph{networks} are leaf-labelled graphs that are used to describe evolutionary histories of species. The \textsc{Tree Containment} problem asks whether a given phylogenetic tree is embedded in a given phylogenetic network. Given a phylogenetic network and a cluster of species, the \textsc{Cluster Containment} problem asks whether the given cluster is a cluster of some phylogenetic tree embedded in the network. Both problems are known to be NP-complete in general. In this article, we consider the restriction of these problems to several well-studied classes of phylogenetic networks. We show that \textsc{Tree Containment} is polynomial-time solvable for normal networks, for binary tree-child networks, and for level-$k$ networks. On the other hand, we show that, even for tree-sibling, time-consistent, regular networks, both \textsc{Tree Containment} and \textsc{Cluster Containment} remain NP-complete.\end{abstract}

\begin{keyword}
algorithms \sep computational complexity \sep phylogenetic trees \sep phylogenetic networks
\end{keyword}

\maketitle

\section{Introduction}

Rooted trees, and more generally digraphs, are widely used to represent evolutionary relationships in biology~\cite{huson06, HusonEtAl2010, Nak10}. Such a digraph typically has a set~$\mathcal{X}$ of labelled leaves (vertices of outdegree 0) which corresponds to the collection of present-day species under study. The arcs in the digraph are directed away from a single `root' vertex which represents the evolutionary ancestor of the species in~$\mathcal{X}$. The remaining vertices of the digraph are usually unlabelled; they represent hypothetical ancestral species, and an arc from~$u$ to~$v$ indicates that ancestral species $u$ contributes directly to the genetic makeup of~$v$.

The simplest phylogenetic networks are trees, and these have traditionally been used  when evolution is described purely by the formation of each new species from an existing one. However, processes of reticulate evolution (in particular the formation of hybrid species, and lateral gene transfer) mean that a new species can have genetic contributions from more than one ancestral species, which is the reason digraphs are increasingly seen as a desirable model in molecular systematics~\cite{doo07}. More precisely, if one considers the evolution of a particular gene, its evolution will generally be described by a tree, but reticulate processes mean that different genes can have different evolutionary tree structures,  that can only be adequately reconciled by fitting them into a phylogenetic network.

This raises a fundamental question in computational phylogenetics: given a tree and a phylogenetic network, both with leaf set~$\mathcal{X}$, is there an efficient algorithm to determine whether the tree `fits inside' the network (in a sense we define precisely below)? In general this question is NP-hard, even for certain restricted classes of phylogenetic networks. However, we show that, for particular classes, there exist polynomial-time algorithms which also provide a recipe for finding an explicit embedding of the tree in the phylogenetic network. In addition, we consider the computational complexity of deciding whether a subset of~$\mathcal{X}$ is a cluster of some tree that sits inside a given network.

The structure of this article is as follows. After giving formal definitions in the next section, we present our polynomial-time algorithms for the `tree containment' problem in Section~\ref{sec:polytime}, discuss the `cluster containment' problem in Section~\ref{sec:cluster}, show the computational intractability of both problems for certain classes of networks in Section~\ref{sec:hard}, and finish with an open problem in Section~\ref{sec:open}.

\section{Definitions}

Consider a set~$\mathcal{X}$ of taxa. A \emph{rooted phylogenetic network} (\emph{network} for short) on~$\mathcal{X}$ is a directed acyclic (simple) graph with a single \emph{root} (vertex with indegree~0), \emph{leaves} (vertices with outdegree~0) bijectively labeled by~$\mathcal{X}$ and no vertices with indegree and outdegree one. We identify each leaf with its label and refer to the directed edges (arcs) as \emph{edges}.

A vertex is called a \emph{reticulation vertex} (\emph{reticulation} for short) if it has indegree at least two and a \emph{tree-vertex} otherwise. An edge is called a \emph{tree-edge} if it ends in a tree-vertex and it is called a \emph{reticulation-edge} otherwise. A \emph{tree-path} is a directed path that contains only tree-edges. A network with no reticulations is said to be a \emph{rooted phylogenetic tree} (\emph{tree} for short). A network is said to be \emph{binary} if each reticulation has indegree two and outdegree one and all of the vertices have outdegree at most two.

Given two vertices~$v_1,v_2$ (of some tree or network), we use~$v_1\leq v_2$ to denote that there is a directed path from~$v_1$ to~$v_2$. If~$v_1\leq v_2$ and $v_1\neq v_2$, we write~$v_1 < v_2$. For a vertex~$v$, we define $cl(v)=\{x\in\mathcal{X}:v\leq x\}$. If~$(u,v)$ is an edge, then we say that~$u$ is a \emph{parent} of~$v$ and~$v$ is a \emph{child} of~$u$. If two vertices~$v_1,v_2$ have a common parent, then they are said to be \emph{siblings}. An edge~$e$ of a network~$N$ is a \emph{cut-edge} if removal of~$e$ disconnects the undirected graph underlying~$N$. A cut-edge is \emph{trivial} if it ends in a leaf.

Given a tree~$T$ and network~$N$, we say that~$N$ \emph{displays}~$T$ if there is a subgraph~$T'$ of~$N$ that is a subdivision of~$T$ (i.e.~$T'$ can be obtained from~$T$ by replacing edges by directed paths). For a tree~$T$ on~$\mathcal{X}$ and a subset~$X'\subseteq\mathcal{X}$, we use~$T|X'$ to denote the tree on~$X'$ that is displayed by~$T$.

These concepts are illustrated in Fig.~\ref{fig:leo}.
\begin{figure}[ht]
  \centering
  \includegraphics[scale=.6]{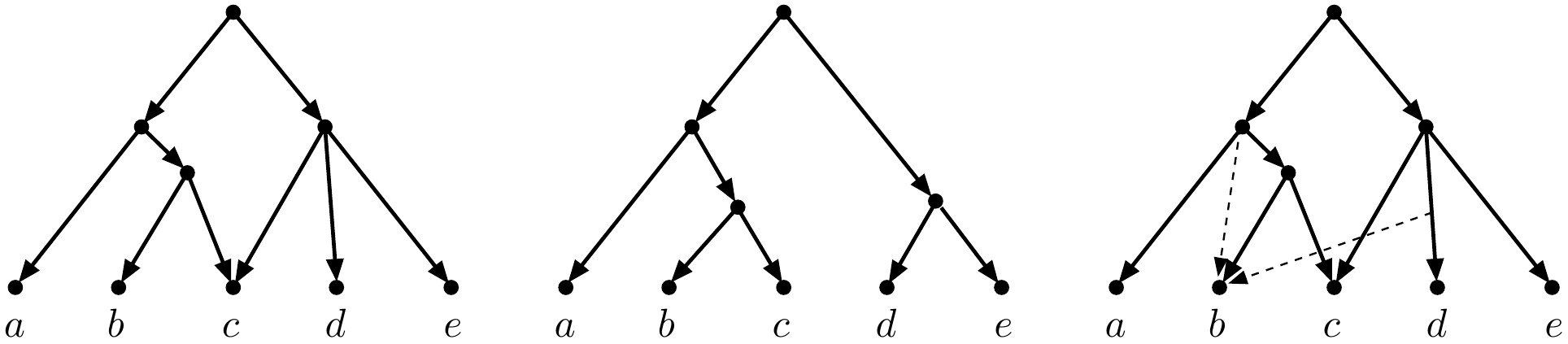}
  \caption{{\em Left}: A rooted phylogenetic network on $\mathcal{X}=\{a,b,c,d,e\}$.  This network is normal, but not binary and has one reticulation vertex.  {\em Middle}: One of the two rooted phylogenetic trees on $\mathcal{X}=\{a,b,c,d,e\}$ that are displayed by this network. {\em Right}: If we add either of the dashed edges, the resulting network is no longer normal; adding the steeper, left-most arrow violates condition (R3) while adding the
  right-most arrow results in a network that violates the tree-child property for the parent of~$b$ and~$c$. }
  \label{fig:leo}
\end{figure}

In the first part of this paper we consider the following fundamental decision problem:

\medskip
\noindent\begin{tabular}{lp{.85\textwidth}}
\multicolumn{2}{l}{\textsc{Tree Containment}} \\
\textit{Instance:} & A set~$\mathcal{X}$ of taxa, a rooted phylogenetic network~$N$ on~$\mathcal{X}$ and a rooted phylogenetic tree~$T$ on~$\mathcal{X}$.\\
\textit{Question:} & Does~$N$ display~$T$? \\
\end{tabular}
\medskip

This problem is known to be NP-complete~\cite{KanjEtAl2008} for the general class of rooted, binary phylogenetic networks.
Here we study the restriction of this problem to some established classes of rooted phylogenetic networks. We then consider the
complexity of a related `cluster containment' problem.

\subsection{Classes of Phylogenetic Networks}
A phylogenetic network $N$ is said to be \emph{regular} if for any two distinct vertices~$u,v$ of~$N$
\begin{enumerate}
\item[(R1)] $cl(u) \neq cl(v)$.
\item[(R2)] $u\leq v$ if and only if~$cl(u)\supset cl(v)$;
\item[(R3)] there is no edge~$(u,v)$ if there is also a directed path from~$u$ to~$v$ of length greater than one.
\end{enumerate}

A phylogenetic network~$N$ is said to be \emph{tree-child} (see e.g. \cite{CardonaEtAl2009}) if each vertex of~$N$ either is a leaf or has a child that is a tree-vertex. It follows immediately that from each vertex of a tree-child network there exists a tree-path to a leaf. A network is said to be \emph{normal} if it is tree-child, has no vertices of outdegree 1, and in addition condition (R3) above holds. It has been shown that any normal network is regular, i.e. that automatically conditions (R1) and (R2) hold~\cite{Willson2009b}.

\section{Polynomial-Time Algorithms}\label{sec:polytime}

This section describes our main results, polynomial-time algorithms for \textsc{Tree Containment} restricted to normal networks and to binary tree-child networks. We first give an algorithm for normal networks in Section~\ref{sec:normal}, because this algorithm is simpler and will be used as a subroutine for the algorithm for binary tree-child networks in Section~\ref{sec:treechild}.

\subsection{Normal Networks}\label{sec:normal}

We show that, given a normal phylogenetic network~$N$ and a phylogenetic tree~$T$, one can decide in polynomial time whether~$N$ displays~$T$. We propose the following algorithm. An example is in Figure~\ref{fig:normalalg}.

\medskip
\noindent\textbf{Algorithm} \textsc{LocateNormal}\\
For each reticulation~$r$ of~$N$, do the following.
\begin{enumerate}
\item Find a leaf~$x_r$ that can be reached from~$r$ by a tree-path.
\item Construct a set~$X'$ consisting of~$x_r$ and, for each parent~$p$ of~$r$, a leaf that can be reached from~$p$ by a tree-path.
\item If~$x_r$ does not have a sibling in~$T|X'$, report that~$N$ does not display~$T$.
\item If~$x_r$ does have a sibling in~$T|X'$, let~$x_s$ be such a sibling and~$p_s$ the reticulation from which~$x_s$ can be reached by a tree-path. Delete all edges entering~$r$ except for $(p_s,r)$.
\end{enumerate}
When there are no reticulations left, the network has been transformed into a phylogenetic tree~$T'$. Check if~$T'$ is a subdivision of~$T$. If it is, report that~$N$ displays~$T$. Otherwise, report that~$N$ does not display~$T$.

\begin{figure}[ht]
  \centering
  \includegraphics[scale=0.5]{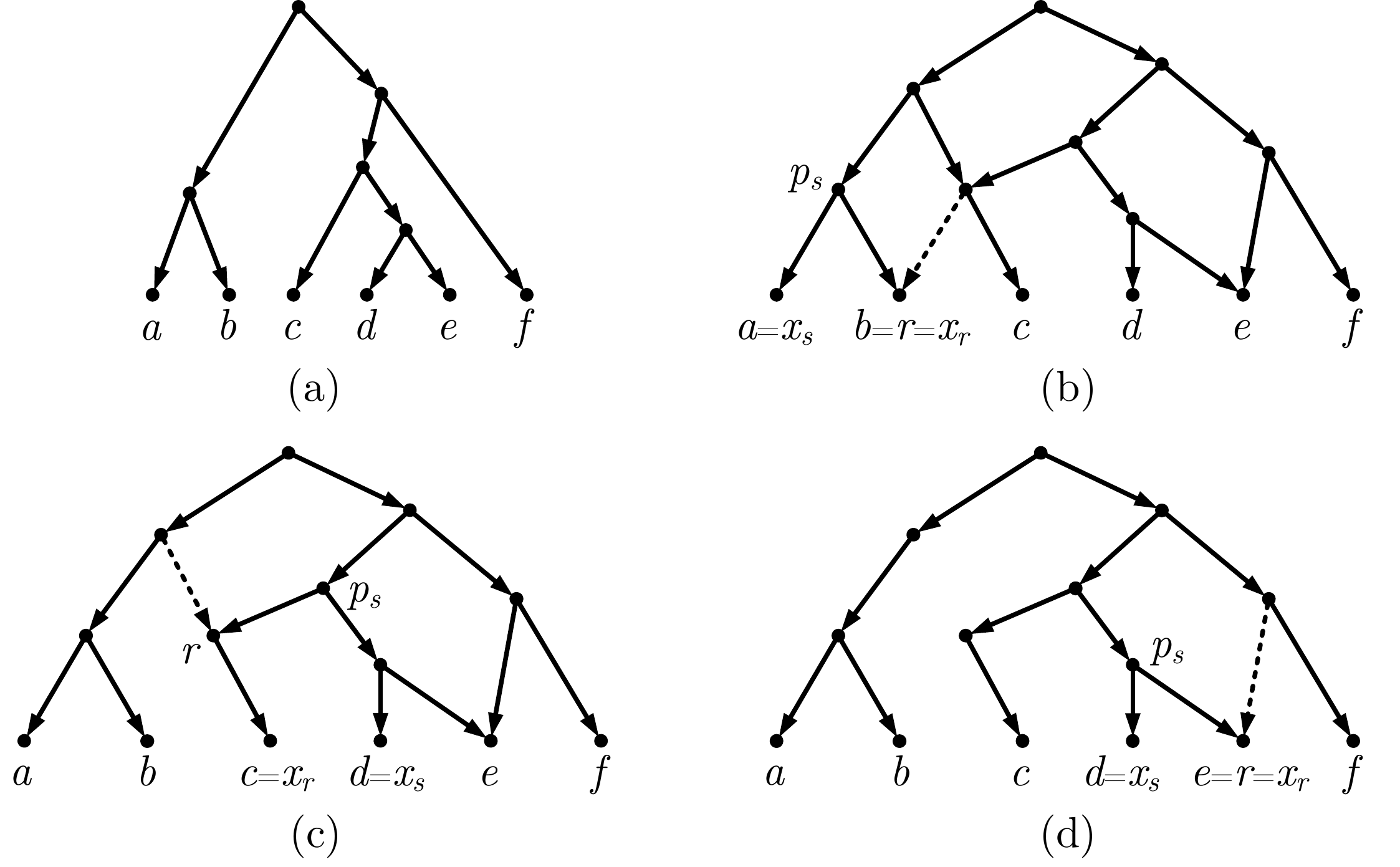}
  \caption{Illustration of algorithm \textsc{LocateNormal} for tree~$T$ in (a) and normal network~$N$ in (b). Dotted edges are removed from~$N$; in~(b) because~$a$ and $b$ are siblings in $T|\{a,b,c\}$, in~(c) because~$c$ and~$d$ are siblings in $T|\{a,c,d\}$, and in~(d) because~$d$ and~$e$ are siblings in $T|\{d,e,f\}$.}
  \label{fig:normalalg}
\end{figure}

\begin{theorem}\label{thm:normal} Given a normal phylogenetic network~$N$ on~$\mathcal{X}$ and a phylogenetic tree~$T$
on~$\mathcal{X}$, algorithm \textsc{LocateNormal} decides in polynomial time whether~$N$ displays~$T$. Moreover, if~$N$ displays~$T$, \textsc{LocateNormal} finds the unique subtree of~$N$ that is a subdivision of~$T$.
\end{theorem}
\begin{proof}
We first show that, for any reticulation~$r$, the leaves in the set~$X'$ described by the algorithm (which exist by the tree-child property) are always distinct. To show this, assume that two such tree-paths end in the same leaf. Since tree-paths cannot combine in a reticulation, the only remaining possibility is that, for two parents~$p_a,p_b$ of~$r$, the tree-path from~$p_a$ to a leaf passes through~$p_b$. However, this implies the existence of a directed path $p_a\rightarrow r$ of length greater than one, contradicting (R3), since there is also an edge~$(p_a,r)$.

We now prove that \textsc{LocateNormal} reports that~$N$ displays~$T$ if and only if this is the case. If the algorithm reports that~$N$ displays~$T$, then it is clear that~$N$ indeed displays~$T$, since the algorithm checks if the constructed subtree~$T'$ of~$N$ is a subdivision of~$T$. Now suppose that~$N$ displays~$T$. Then there exists a subtree~$T'$ of~$N$ that is a subdivision of~$T$. Take any reticulation~$r$ and~$X'$ as in the algorithm. Observe that the root of~$T'$ is the root of~$N$ by property~(R1). Thus, by this observation and the fact that there is a tree-path from~$r$ to a leaf,~$r$ and exactly one of the edges entering~$r$ are included in~$T'$. Let~$(p',r)$ be the edge entering~$r$ that is included in~$T'$ and~$x'$ the leaf in~$X'$ that can be reached from~$p'$ by a tree-path. Then~$x'$ is a sibling of~$x_r$ in~$T'|X'$ (which is isomorphic to $T|X'$) and hence $p'=p_s$. It follows that \textsc{LocateNormal} deletes all reticulation-edges not included in~$T'$, thus reconstructing~$T'$. Moreover, since there is no choice which reticulation-edges to delete, the subtree~$T'$ is unique.

The running-time is clearly polynomial in the size of the input. Moreover, the number of vertices of a normal network grows at most quadratically with the number of leaves~\cite{Willson2009b}. Hence, the running-time is also polynomial in the number of leaves. \qed
\end{proof}

\subsection{Tree-Child Networks}\label{sec:treechild}

This subsection shows that, given a binary tree-child phylogenetic network~$N$ and a phylogenetic tree~$T$, one can decide in polynomial time whether~$N$ displays~$T$. We partition the reticulations of~$N$ into four types.\\
\textbf{Type I.} There is no tree-path between the parents of the reticulation.\\
\textbf{Type II.} There is a singe edge connecting the parents of the reticulation.\\
\textbf{Type III.} There is a tree-path between the parents of the reticulation that contains a vertex that is the parent of two tree-vertices.\\
\textbf{Type IV.} There is a tree-path between the parents of the reticulation that contains at least one internal vertex but no vertex that is the parent of two tree-vertices.

Note that in a normal network all reticulations are of Type~I, because of restriction (R3). Thus, the difficulty in generalizing the algorithm for normal networks in Section~\ref{sec:normal} is that the algorithm for tree-child networks below also has to take reticulations of Types~II, III and (especially) IV into account.

\medskip
\noindent\textbf{Algorithm} \textsc{LocateTreeChild}\\
Repeat the following four steps until none is applicable.\\
\noindent\textbf{Step 1.} If there is a reticulation of Type I, proceed as for normal networks.\\
\noindent\textbf{Step 2.} If there is a reticulation of Type II, choose one of its two incoming edges arbitrarily and remove it.\\
\noindent\textbf{Step 3.} If there is a reticulation~$r$ of Type III, do the following.
\begin{enumerate}
\item[3.1] Let~$p_a$ and~$p_b$ be the parents of~$r$ such that~$p_a<p_b$. Pick a vertex~$v$ that lies on the path from~$p_a$ to~$p_b$ and is the parent of two tree-vertices. Let~$c$ be the child of~$v$ that does not lie on the path between the parents of~$r$. Let~$x_r,x_b$ and~$x_c$ be leaves that can be reached from~$r,p_b$ and~$c$ respectively by tree-paths.
\item[3.2] If~$x_b$ and~$x_r$ are siblings in~$T|\{x_b,x_r,x_c\}$, remove the edge $(p_a,r)$ from~$N$.
\item[3.3] If~$x_b$ and~$x_c$ are siblings in~$T|\{x_b,x_r,x_c\}$, remove the edge $(p_b,r)$ from~$N$.
\item[3.4] Otherwise, report that~$N$ does not display~$T$.
\end{enumerate}
\noindent\textbf{Step 4.} If there is a reticulation of Type IV, do the following.
\begin{enumerate}
\item[4.1] If~$N$ contains a nontrivial cut-edge~$(u,v)$, let~$C$ be the set of taxa reachable from~$v$. Find a vertex~$v'$ of~$T$ with $C=cl(v')$ (if there is no such vertex, report that~$N$ does not display~$T$). Construct~$T_c$ from~$T$ by deleting all descendants of~$v'$ and labeling~$v'$ by a new taxon~${x_c\notin\mathcal{X}}$. Construct~$N_c$ from~$N$ by deleting all descendants of~$v$ and labeling~$v$ by~$x_c$. Construct $N|C$ by restricting~$N$ to the vertices and edges reachable from~$v$. Run the algorithm recursively for $(N_c,T_c)$ and $(N|C,T|C)$. If~$N_c$ displays~$T_c$ and~$N|C$ displays~$T|C$, report that~$N$ displays~$T$. Otherwise, report that~$N$ does not display~$T$.
\item[4.2] Contract all vertices with indegree and outdegree 1. Find a leaf~$x$ that has a reticulation as sibling. Construct a tree-path~$P$ as follows. Initialise~$P$ as~$x$ and, as long as the first vertex~$v$ of~$P$ has a reticulation as sibling, add the parent of~$v$ to~$P$. Construct a set~$X'$ by including~$x$ and, for each reticulation that has a parent on~$P$, a leaf that can be reached from this reticulation by a tree-path.
\item[4.3] If~$x$ does not have a sibling in~$T|X'$, report that~$N$ does not display~$T$.
\item[4.4] If~$x$ does have a sibling~$x_s$ in~$T|X'$, let~$r'$ be the reticulation from which~$x_s$ can be reached by a tree-path. Let~$p'_a$ and~$p'_b$ be the parents of~$r'$ such that $p'_a<p'_b$. Remove the edge $(p'_a,r')$ from~$N$.
\end{enumerate}

When there are no more reticulations left, check if the resulting phylogenetic tree~$T'$ is a subdivision of~$T$. If it is, report that~$N$ displays~$T$. Otherwise, report that~$N$ does not display~$T$.

\begin{theorem}\label{thm:treechild} Given a binary tree-child phylogenetic network~$N$ on~$\mathcal{X}$ and a phylogenetic tree~$T$ on~$\mathcal{X}$, algorithm \textsc{LocateTreeChild} decides in polynomial time whether~$N$ displays~$T$.
\end{theorem}
\begin{proof}
If there exist reticulations of Type~I or~III, it can be shown analogously to the proof of Theorem~\ref{thm:normal} that Steps~1 and~3 of \textsc{LocateTreeChild} either resolve these reticulations in the only possible way, or conclude that~$N$ does not display~$T$. It is furthermore clear that Step~2 correctly deals with reticulations of Type~II (because~$N$ is binary) and Step~4.1 with nontrivial cut-edges. Thus, from now on, we assume that all reticulations are of Type~IV and that there are no nontrivial cut-edges.

We first show that a leaf~$x$ as described by Step~$4.2$ indeed exists, i.e. we show that there is a leaf that has a sibling which is a reticulation. Consider a tree-path~$Q$ of maximum length. Because~$Q$ is a tree-path, it ends in a tree-vertex~$x$, which must be a leaf because~$N$ is tree-child and~$Q$ is of maximum length. The single parent of~$x$ is a tree-vertex because, if it were a reticulation,~$Q$ would consist of just~$x$ but, by the existence of a Type~IV reticulation, there exists a tree-path of length at least two. Thus, leaf~$x$ has a sibling, which has to be a reticulation because otherwise it would be possible to extend~$Q$ to a longer tree-path. This shows the existence of a leaf~$x$ as described in Step~$4.2$.

We now show that the algorithm reports that~$N$ displays~$T$ precisely when this is the case. The algorithm checks if the constructed subgraph~$T'$ is a subdivision of~$T$, so it will clearly never report that~$N$ displays~$T$ if this is not the case. So assume that~$N$ does display~$T$. It remains to show that the algorithm will construct some subdivision of~$T$ in~$N$. Observe that, apart from~$x$, all vertices on the path~$P$ have two children: one reticulation and one tree-vertex (which lies on~$P$). Let~$R$ be the set of all reticulations with a parent on~$P$. Thus,~$X'$ consists of~$x$ and, for each reticulation in~$R$, a leaf that can be reached from this reticulation by a tree-path. By the choice of~$x$, the path~$P$ contains at least one other vertex, and so the set~$X'$ contains at least two leaves.

We claim that all parents of reticulations in~$R$ lie on~$P$. Assume that this is not the case. Then there is a reticulation~$\hat{r}$ with two parents~$\hat{p}_a,\hat{p}_b$ such that $\hat{p}_a<\hat{p}_b$ and one of~$\hat{p}_a,\hat{p}_b$ is not on~$P$ while the other one is. Since~$P$ and the path $\hat{p}_a<\hat{p}_b$ both contain only tree-edges, we must have that~$\hat{p}_a$ is on~$P$ while~$\hat{p}_b$ is not. However, then the path~$\hat{p}_a\rightarrow\hat{p}_b$ contains a vertex that is the parent of two tree-vertices, which means that~$\hat{r}$ is of Type~III. This is a contradiction because the algorithm has already resolved all reticulations of Type~III in Step~3.

Thus, apart from~$x$, $P$ consists precisely of the parents of reticulations in~$R$. Since~$N$ displays~$T$, there exists a subtree~$T'$ of~$N$ that is a subdivision of~$T$. We next claim that~$x$ has a sibling in~$T'|X'$. Let~$v_s$ be the last vertex on the path~$P$ for which the reticulation-edge leaving~$v_s$ is included in~$T'$ (such a vertex exists because one of the edges entering the sibling of~$x$ is included in~$T'$). Let~$x_s$ be the leaf in~$X'$ that can be reached from the reticulation-child of~$v_s$ by a tree-path. Clearly,~$x_s$ and~$x$ are siblings in~$T'|X'$ (which is isomorphic to~$T|X'$). Thus, the algorithm removes the edge $(p'_a,r')$ from~$N$, with~$r'$ the reticulation-child of~$v_s$ and~$p'_a$ equal to either~$v_s$ or the other parent of~$r'$.

First suppose that~$p'_a$ is not equal to~$v_s$. Thus, $v_s=p'_b$ and~$p'_a$ is the other parent of~$r'$. In this case, the algorithm correctly removes the edge $(p'_a,r')$ that is not in~$T'$ and we are done.

Now assume that~$p'_a=v_s$ and thus that the algorithm removes the edge $(p'_a,r')$ that is in~$T'$. We replace in~$T'$ the edge $(p'_a,r')$ by the edge $(p'_b,r')$. The resulting subgraph of~$N$ is again a subdivision of~$T$, because none of the reticulation-edges leaving vertices below~$p'_a=v_s$ were included in~$T'$ (by the choice of~$v_s$). See Figure~\ref{fig:treechildproof} for an illustration. We conclude that the algorithm correctly constructs a subdivision of~$T$ in~$N$.

The running-time is clearly polynomial in the size of the input. Moreover, the number of reticulations in a tree-child network is at most~$n=|\mathcal{X}|$ and it can be shown by induction on the number~$n_r$ of reticulations that the number of edges of a binary tree-child network is at most~$2n+3n_r-2\leq 5n-2$. Hence, the running-time is also polynomial in the number of leaves. \qed
\end{proof}

\begin{figure}[ht]
  \centering
  \includegraphics[scale=.5]{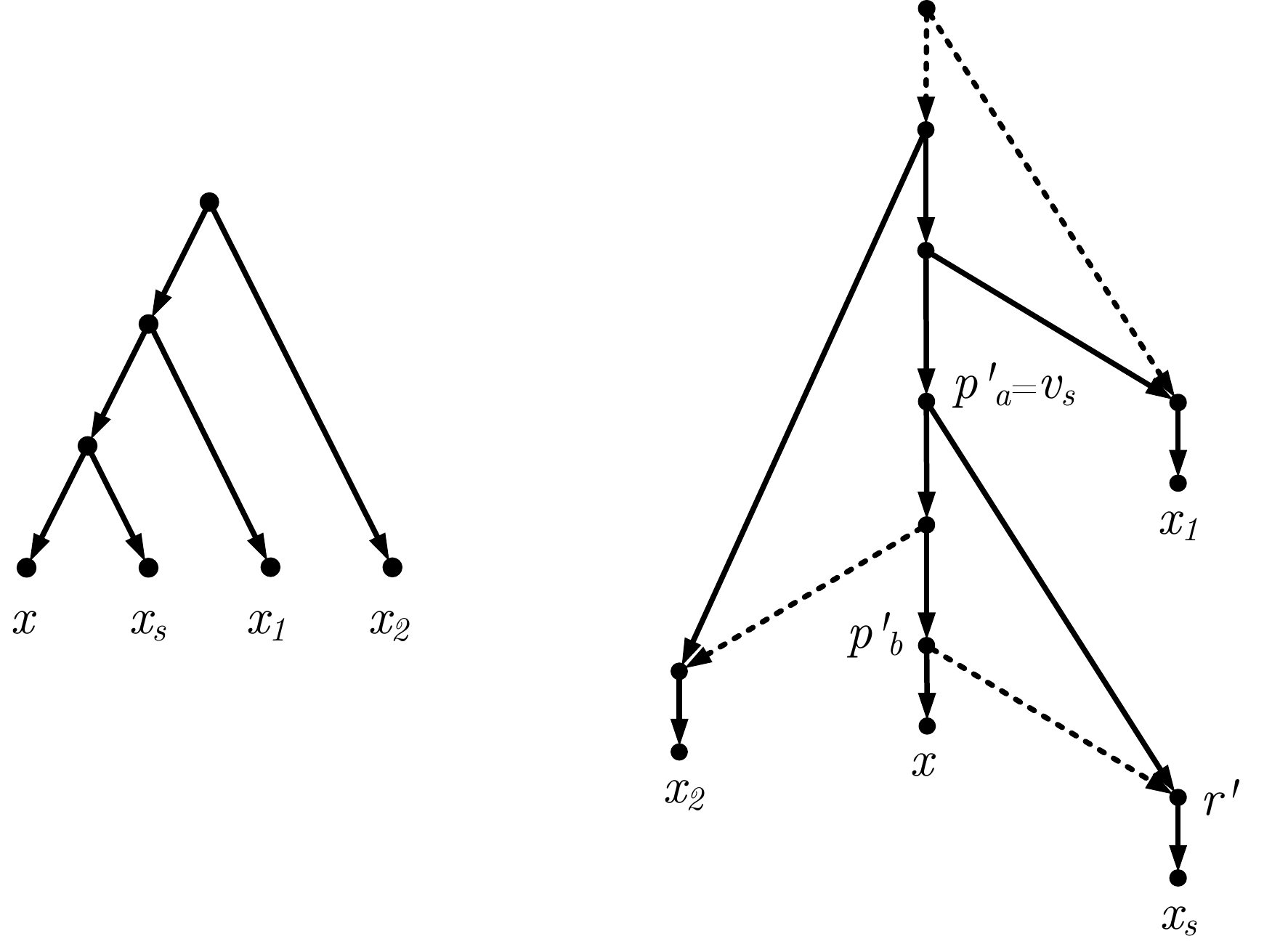}
  \caption{Illustration of the last case in the proof of Theorem~\ref{thm:treechild}. A subdivision of the tree~$T$ in the network~$N$ is indicated by solid lines. Replacing edge $(p'_a,r')$ by $(p'_b,r')$ gives another subdivision of~$T$ in~$N$, because $(p'_a,r')$ is the last edge leaving the path~$P$ (the vertical path) that is used by the subdivision.}
  \label{fig:treechildproof}
\end{figure}

\section{Cluster Containment}\label{sec:cluster}

A \emph{cluster} is a strict subset of~$\mathcal{X}$. There are two different ways of seeing clusters in networks. Let~$v$ be a vertex of a network~$N$ on~$X$. The \emph{hardwired cluster} of~$v$, denoted~$cl(v)$, is the set containing all taxa that can be reached from~$v$. A cluster is said to be a hardwired cluster of~$N$ if it is the hardwired cluster of some vertex of~$N$. A cluster~$C\subset\mathcal{X}$ is said to be a \emph{softwired cluster} of~$v$ if~$C$ equals the set of all taxa that can be reached from~$v$ when, for each reticulation~$r$, exactly one incoming edge of~$r$ is ``switched on'' and the other incoming edges of~$r$ are ``switched off''. Notice that each vertex has exactly one hardwired cluster and one or more softwired clusters. We say that~$C$ is a \emph{softwired cluster} of a network~$N$ if it is a softwired cluster of some vertex of~$N$. The softwired clusters of~$N$ can elegantly be characterized as follows. A cluster~$C\subset\mathcal{X}$ is a softwired cluster of~$N$ if and only if there exists a tree~$T$ such that~$C$ is a hardwired cluster of~$T$ and~$T$ is displayed by~$N$. We are interested in the following decision problem.

\medskip
\noindent\begin{tabular}{lp{.85\textwidth}}
\multicolumn{2}{l}{\textsc{Cluster Containment}} \\
\textit{Instance:} & A set~$\mathcal{X}$ of taxa, a rooted phylogenetic network~$N$ on~$X$ and a cluster~${C\subset\mathcal{X}}$.\\
\textit{Question:} & Is~$C$ a softwired cluster of~$N$? \\
\end{tabular}
\medskip

This problem is known to be NP-complete~\cite{KanjEtAl2008} for general binary phylogenetic networks. However, it has a polynomial-time algorithmic solution if we restrict~$N$ to tree-child networks~\cite{HusonEtAl2010}. We will show below (Theorem \ref{thm:reghard}) that, if we extend the class of normal networks to regular networks, the problem \textsc{Cluster Containment} becomes NP-hard, even if we add further structural restrictions to this class of networks. But first we describe yet another class of networks for which both problems --\textsc{Tree Containment} and \textsc{Cluster Containment} -- have polynomial-time algorithms.

\subsection{Level-$k$ Networks}

We finish this section by observing that both \textsc{Tree Containment} and \textsc{Cluster Containment} are polynomial-time solvable for the class of binary level-$k$ networks. A binary network is \emph{biconnected} if it contains no cut-edges. A biconnected subgraph~$B$ of a binary network~$N$ is said to be a \emph{biconnected component} if there is no biconnected subgraph~$B' \neq B$ of $N$ that contains~$B$. A binary phylogenetic network is a \emph{${\mbox{level-}k}$ network} if each biconnected component has at most~$k$ reticulations.

\begin{observation}\label{obs:levelk}
\textsc{Tree Containment} and \textsc{Cluster Containment} are polynomial-time solvable when restricted to binary level-$k$ networks, for any fixed~$k$.
\end{observation}
\begin{proof} We may assume that the network contains no nontrivial cut-edges, because Step~4.1 of algorithm \textsc{LocateTreeChild} can be applied until there are no nontrivial cut-edges left. A level-$k$ network with no nontrivial cut-edges contains at most~$k$ reticulations. Thus, we can loop through all~$2^k$ ways of selecting one incoming edge for each of the~$k$ reticulations. For each of the resulting trees, we check whether it is a subdivision of the input tree (in case of the \textsc{Tree Containment} problem) or check whether the input cluster is a cluster of the obtained tree (in case of the \textsc{Cluster Containment} problem).
\qed
\end{proof}

\section{NP-Completeness for Tree-Sibling, Time-Consistent, Regular Networks}\label{sec:hard}

A phylogenetic network is said to be \emph{tree-sibling} (see e.g.~\cite{CardonaEtAl2008}) if each reticulation has a sibling that is a tree-vertex. A phylogenetic network~$N$ is said to be \emph{time-consistent} (see e.g.~\cite{BaroniEtAl2006,LinzEtAl}) if it is possible to assign each vertex~$v$ of~$N$ a ``time stamp'' $t(v)\in\mathbb{R}^{\geq 0}$ such that for each edge~$(u,v)$ of~$N$:
\begin{enumerate}
\item[(TC1)] $t(u) < t(v)$ if $(u,v)$ is a tree-edge and
\item[(TC2)] $t(u) = t(v)$ if $(u,v)$ is a reticulation-edge.
\end{enumerate}

\begin{theorem}\label{thm:reghard} \textsc{Tree Containment} and \textsc{Cluster Containment} are both NP-complete when restricted to tree-sibling, time-consistent, regular phylogenetic networks.
\end{theorem}
\begin{proof}
We reduce from the \textsc{Tree Containment} and \textsc{Cluster Containment} problems on general networks, which were shown to be NP-complete by Kanj et al.~\cite{KanjEtAl2008}.

Let~$\mathcal{X}$ be a set of taxa, and let~$N$ and~$T$ be a network and tree on~$\mathcal{X}$, respectively. We will modify~$N$ to a tree-sibling, time-consistent, regular network~$N'$ on a set of taxa~$\mathcal{X}'\supseteq\mathcal{X}$ and show that a cluster~$C\subset\mathcal{X}$ is a softwired cluster of~$N'$ if and only if it is a softwired cluster of~$N$. Since this modification can be carried out in polynomial time and the size of~$N'$ is polynomial in the size of~$N$, this shows NP-hardness of \textsc{Cluster Containment} on tree-sibling, time-consistent, regular networks. Furthermore, we will (in polynomial time) modify~$T$ to a tree~$T'$ on~$\mathcal{X}'$ and show that~$N'$ displays~$T'$ if and only if~$N$ displays~$T$, thus showing NP-hardness of \textsc{Tree Containment} on this restricted class of networks.

The construction of~$N'$ and~$T'$ relies on repeatedly applying the following operation. For a vertex~$v$ of~$N$, define \textsc{HangLeaves}$(v)$ as making the following changes to~$N$,~$T$ and~$\mathcal{X}$. We add two new taxa~$x,x'$ to~$\mathcal{X}$. Let~$r$ be the root of~$N$. We add to~$N$ leaves~$x,x'$, a new root~$r'$, an internal vertex~$p$ and edges~$(r',r),(r',p),(p,x'),(p,x)$ and~$(v,x)$. Let~$r_T$ be the root of~$T$. We add to~$T$ the leaves~$x,x'$, a new root~$r_T'$, an internal vertex~$p_T$ and edges~$(r_T',r_T),(r_T',p_T),(p_T,x)$ and~$(p_T,x')$. See Figure~\ref{fig:modification}.

\begin{figure}[ht]\centering
\includegraphics[scale=.5]{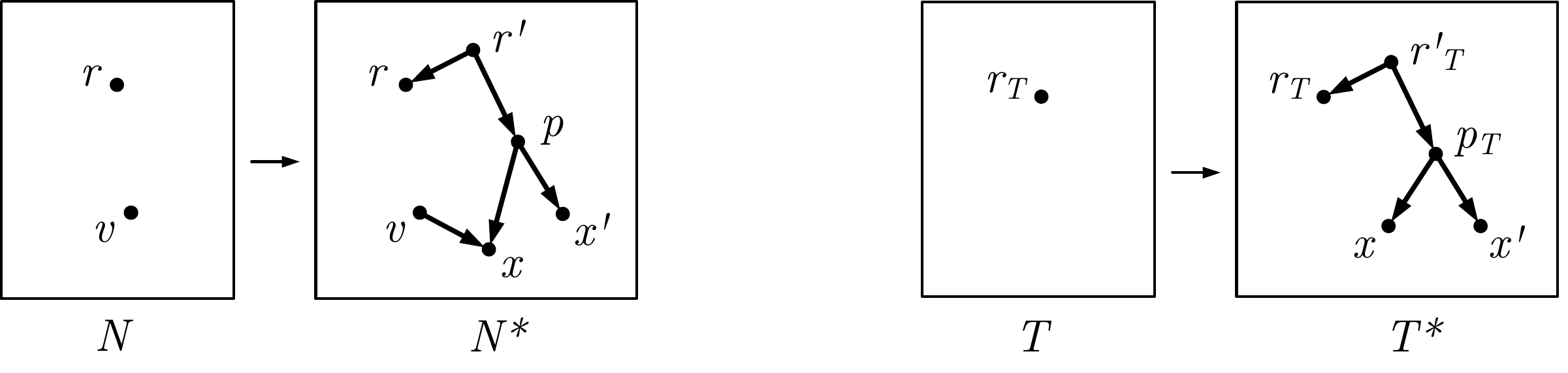}
\caption{How \textsc{HangLeaves}$(v)$ modifies~$N$ and~$T$. Vertices~$r$ and~$r_T$ are the roots of~$N$ and~$T$ respectively.}\label{fig:modification}
\end{figure}

We now describe how we transform~$N$ into a tree-sibling, time-consistent, regular network by repeated applications of \textsc{HangLeaves}. First, we make the network regular by doing the following for each pair~$u,v$ of distinct vertices of~$N$.
\begin{enumerate}
\item If~$cl(u)=cl(v)$, then apply \textsc{HangLeaves}$(u)$ and \textsc{HangLeaves}$(v)$.
\item If~$cl(u)\supset cl(v)$ but there is no directed path from~$u$ to~$v$, then apply \textsc{HangLeaves}$(v)$.
\item If there exist two distinct directed paths from~$u$ to~$v$, one of which is an edge, then subdivide this edge with a single vertex~$w$ and apply \textsc{HangLeaves}$(w)$.
\end{enumerate}
For any two vertices~$u,v$ of~$N$, operations~1,~2 and~3 make sure that properties~(R1),~(R2) and~(R3) (respectively) of a regular network are satisfied. Furthermore, by the definition of \textsc{HangLeaves}, these properties are also satisfied for newly added vertices. It follows that the obtained network is regular. Call this network~$N_r$.

The next step is to make the network time-consistent. For each reticulation-edge $(u,v)$ of~$N_r$, subdivide $(u,v)$ by a new vertex~$w$ and apply \textsc{HangLeaves}$(w)$. Let~$N_{rtc}$ be the resulting phylogenetic network.

We claim that~$N_{rtc}$ is time-consistent. Let~$t:V(N_r)\rightarrow\mathbb{N}$ be a labeling of the vertices of~$N_r$ such that ${t(u)<t(v)}$ for each edge $(u,v)$. This is possible because~$N_r$ is acyclic. We specify a label~$t'(v)$ (the time-stamp) for each vertex~$v$ of~$N_{rtc}$ as follows. Each vertex of~$N_{rtc}$ that is also a vertex of~$N_r$ gets the same label as in~$N_r$, i.e. $t'(v)=t(v)$ for all~$v\in V(N_r)$. Now consider a reticulation-edge $(u,v)$ of~$N_r$. Such an edge corresponds to two edges $(u,w)$ and $(w,v)$ of~$N_{rtc}$. Now label vertex~$w$ the same as vertex~$v$, i.e. $t'(w)=t'(v)=t(v)$. Observe that $(w,v)$ is a reticulation-edge and now satisfies restriction (TC2) of time-consistency. Furthermore, we have $t'(u)<t'(w)$ and so $(u,w)$, which is a tree-edge, satisfies restriction (TC1) of time-consistency. It remains to label the vertices that have been added by \textsc{HangLeaves}. This can easily be done in such a way that the restrictions of time-consistency are satisfied. Namely, we give~$x$ and~$p$ the same label as~$v$, give~$r'$ any label that's smaller than $t'(r)$ and~$x'$ any label that is greater than $t'(p)$ (processing vertices in the same order as in which they have been added by applications of \textsc{HangLeaves}).

Finally, we make the network tree-sibling. For each reticulation~$r$ of~$N_{rtc}$, that has not been added by \textsc{HangLeaves}, we do the following. Observe that, as a result of the modifications that made~$N_{rtc}$ time-consistent,~$r$ has two siblings, both of which are reticulations added by two different applications of \textsc{HangLeaves}. Pick any of the two siblings and call it~$x$. Let~$v$ be the common parent of~$x$ and~$r$. Subdivide edge~$(v,x)$ by a new vertex~$w$ and apply \textsc{HangLeaves}$(w)$. Reticulation~$r$ now has a sibling that is a tree-vertex, namely~$w$. Moreover, all reticulations~$x$ added by applications of \textsc{HangLeaves} have a sibling~$x'$ that is a tree-vertex. Hence, the resulting network is tree-sibling. Let~$N'$ be this resulting network and~$T'$ the resulting tree.

We claim that~$N'$ is not only tree-sibling, but also still regular and time-consistent. To see that~$N'$ is regular, observe that it has been obtained from the regular network~$N_r$ by repeatedly subdividing an edge by a new vertex~$w$ and applying \textsc{HangLeaves}$(w)$. It can easily be checked that a regular network remains regular after such a modification. To see that~$N'$ is also time-consistent, observe that it has been constructed from time-consistent network~$N_{rtc}$ by repeatedly subdividing an edge~$(v,x)$ by a new vertex~$w$ and applying \textsc{HangLeaves}$(w)$. Using that~$x$ is a leaf added by \textsc{HangLeaves}, it can easily be checked that a time-consistent network remains time-consistent after such a modification. Thus,~$N'$ is a tree-sibling, time-consistent, regular network.

It remains to show that (i)~a cluster~$C\subset\mathcal{X}$ is a softwired cluster of~$N'$ if and only if it is a softwired cluster of~$N$ and (ii)~$N'$ displays~$T'$ if and only if~$N$ displays~$T$. The crux to showing these things is that~$N'$ and~$T'$ have been obtained from~$N$ and~$T$ by subdividing edges and applying \textsc{HangLeaves}. By this observation,~(i) is clear.

To see~(ii), let~$N^*$, $T^*$ and~$\mathcal{X}^*$ be the result of a single application of \textsc{HangLeaves}$(v)$ to $N,T,\mathcal{X}$. We claim that~$N^*$ displays~$T^*$ if and only if~$N$ displays~$T$. First note that an embedding of~$T$ in~$N$ can easily be extended to an embedding of~$T^*$ in~$N^*$ by adding to the embedding all new vertices, a path from~$r'$ to the root of the embedding and edges $(r',p)$, $(p,x')$, $(p,x)$. Now consider an embedding of~$T^*$ in~$N^*$. Since~$x$ and~$x'$ are siblings in~$T^*$, this embedding necessarily contains the newly added vertices and edges except for the edge $(v,x)$. Thus, the restriction of the embedding of~$T^*$ in~$N^*$ to an embedding of~$T$ in~$N^*$ does not contain any of the newly added vertices and edges and is thus an embedding of~$T$ in~$N$. Thus,~$N^*$ displays~$T^*$ if and only if~$N$ displays~$T$. By recursively applying this argument, it follows that~$N'$ displays~$T'$ if and only if~$N$ displays~$T$.
\qed
\end{proof}

\section{Open Problem}\label{sec:open}
For a vertex~$v$ and a leaf~$x$ of some phylogenetic network, we say that~$v$ is a \emph{stable ancestor} of~$x$ if all directed paths from the root to~$x$ pass through~$v$. A network is said to be \emph{reticulation-visible} if each reticulation is a stable ancestor of some leaf. Recently, it was shown that \textsc{Cluster Containment} is polynomial-time solvable for reticulation-visible networks~\cite{HusonEtAl2010}. This class of networks contains, but is more general than, the class of tree-child networks. Thus, the tantalizing question remaining open after this work is whether \textsc{Tree Containment} is also polynomial-time solvable for reticulation-visible networks.

Note that \textsc{Tree Containment} cannot simply be solved by checking if each cluster of the input tree~$T$ is a softwired cluster of the input network~$N$ (using an algorithm for \textsc{Cluster Containment}). This approach fails because, even if all clusters of~$T$ are softwired clusters of~$N$, and even if~$N$ is reticulation-visible, it might be that~$N$ does not display~$T$, see~\cite{vanIerselKelk2010}. Thus, there is no obvious reduction from \textsc{Cluster Containment} to \textsc{Tree Containment} or vice versa.

\end{document}